\documentclass[12pt]{article}
\usepackage[margin=1in]{geometry}
\usepackage{setspace}
\usepackage{graphicx}
\usepackage{amssymb}
\usepackage{amsmath}
\usepackage{xcolor}
\usepackage{amsthm}
\usepackage{hyperref}
\usepackage{booktabs}
\usepackage{tabularx}
\usepackage{comment}
\usepackage{authblk}

\newtheorem{proposition}{Proposition}

\title{Towards welfare-oriented recommendations in \\ activity-travel behavior}
\date{}

\author[1,2,*]{Ekin Uğurel}
\author[2]{Takahiro Yabe}

\affil[1]{Department of Industrial and Systems Engineering, University of Houston, Houston, TX, USA}
\affil[2]{Department of Technology Management and Innovation, New York University, New York, NY, USA}

\begin{document}
\maketitle

\noindent\footnotesize{*Corresponding author: \href{mailto:eugurel@central.uh.edu}{eugurel@central.uh.edu}}

\begin{abstract}
    While mainstream recommender systems (RS) rely on diverse heuristics to rank alternatives, they generally lack a principled account of user welfare (i.e., whether accepting the recommendation will leave the user better off than other alternatives). The problem is particularly acute in activity-based travel behavior, where users incur costs they cannot recoup (i.e., energy, time) regardless of eventual satisfaction. As a result, existing systems may recommend options based on popularity or collaborative filtering, but may still leave users worse off than nearby or self-selected alternatives. We address this gap by introducing a welfare-oriented framework for activity recommendation that evaluates suggestions in terms of net utility, defined as experienced benefit minus travel costs. Specifically, we formalize two operational decision criteria: Positive Utility Probability (PUP) recommends only when the probability of non-negative net utility exceeds a threshold, while Regret Minimization (RM) recommends only when expected regret relative to the user’s best organic alternative falls below a tolerance level. To evaluate these criteria, we develop an agent-based simulation in which heterogeneous synthetic travelers interact with multiple RS over time in a spatial environment with realistic travel costs, congestion, and behavioral feedback loops. This framework enables controlled counterfactual evaluations, and offers a practical foundation for designing RS that treat user welfare as a primary objective rather than an incidental byproduct. The simulation code can be found \href{https://github.com/ekinugurel/welfare_oriented_rec_sys/blob/main/paper_figures_and_tables.ipynb}{here}.
\end{abstract}

\section{Introduction}

\label{sec:intro}

Recommender systems (RS) have become pervasive intermediaries in everyday decision-making, from media consumption to restaurant selection to route planning \cite{chaudhari_comprehensive_2020, pedreschi_human-ai_2025}. Because these systems exert such significant influence over daily routines, their design increasingly relies on signals intended to align suggestions with a user's actual well-being. This is particularly evident in location-based recommendations, where platforms incorporate multi-dimensional data to optimize the user experience. For instance, Google Maps ranks venues by balancing proximity with quality-based prominence \cite{google_local_ranking}, OpenTable filters for both culinary preference and real-time availability \cite{kuo2015contextual}, and platforms like Yelp or TripAdvisor utilize collaborative filtering to surface popular, high-value options \cite{nilashi2018travelers}.

The difficulty is that these signals are typically integrated in an ad hoc manner. Distance, for instance, is treated "as the crow flies" rather than as a component of a generalized travel cost that accounts for mode-specific travel times, monetary expenditures, or the opportunity cost of time spent in transit \cite{quarmby1967choice}. Popularity and aggregate ratings reflect revealed preferences (RPs) averaged over many users, but they ignore heterogeneity in tastes, decision-making styles, and situational context. Collaborative filtering captures behavioral similarity, but behavioral similarity is not welfare similarity: two users who tend to visit the same places may nonetheless differ substantially in how they evaluate the trade-off between travel cost and activity enjoyment \cite{vanwee_meta-theory_2025}. Moreover, decades of behavioral research have shown that RPs are themselves imperfect proxies for welfare. People are susceptible to temptation \cite{anwar_recommendation_2025}, exhibit loss aversion and reference dependence \cite{kahneman_prospect_1979}, satisfice under cognitive load \cite{caplin2011search, payne1993adaptive}, and fall into habitual patterns that persist even when better alternatives exist \cite{vanwee_meta-theory_2025, di2022information}. A system that ranks options using these RP signals without a welfare-theoretic foundation inherits these biases and may amplify them through feedback loops \cite{mansoury2020feedback, chaney_how_2018, thorburn_societal_2024}.



This gap between heuristic and principled integration matters especially for \emph{activity} recommendations (i.e., what to do and where to go), which impose a travel cost regardless of eventual satisfaction. Time spent commuting, fuel, congestion, and missed alternatives are real costs that standard RS metrics ignore. When a system suggests a distant option, it asks the user to invest time and money upfront; if the experience disappoints, that cost is unrecoverable. For example, a user may follow a social media influencer's recommendation to a highly rated restaurant across town, endure traffic, and leave wishing they had gone to a nearby cafe. The system ranked the restaurant highly because it was popular and not prohibitively far, but it did not estimate whether the net utility (enjoyment minus travel cost) actually exceeded the user’s next-best local option. This gap between what systems \emph{could} evaluate and what they actually evaluate is the central concern of this paper.

Van Wee and Mokhtarian \cite{vanwee_meta-theory_2025} outline a meta-theory for travel-related choices (MTTC) identifying interacting building blocks, which include personal characteristics, the decision-making paradigm (e.g., utility maximization, regret minimization), and contextual factors. All of these govern how travelers evaluate alternatives. Individuals differ in how they evaluate tradeoffs, with some placing substantially more weight on potential losses or disappointments than on equivalent gains \cite{kahneman1984choices, tversky1991loss}. Bridging welfare economics from this tradition and personalization from the RS tradition opens a productive design space: rather than asking ``which venue should we rank first?’’, we can ask a normatively richer question: ``will following this recommendation leave the user better off than their next-best alternative?’’

In this paper, we argue that when the answer to the latter question is uncertain or likely negative, the recommender should either (1) \emph{abstain} from making the recommendation and let the user make an organic choice, or (2) prioritize more conservative recommendations that have a lower probability of resulting in a net negative experience for the user. We define user welfare at the individual level through the net utility of the recommended outing, which incorporates both the experienced benefit of the activity and the costs required to realize it (i.e., travel time, money, and effort). 

To formalize this intuition, we define two welfare criteria that a recommender can evaluate before issuing a suggestion. The first, \emph{Positive Utility Probability} (PUP), requires the system to estimate whether the user will derive non-negative net utility from the recommended activity, accounting for both the activity experience and the travel cost of reaching it. Only recommendations exceeding a confidence threshold $\alpha$ are issued. The second, \emph{Regret Minimization} (RM), compares the expected utility of the recommended activity against the best alternative the user would have chosen independently. Only recommendations whose expected regret falls below a tolerance $\epsilon$ are issued. Both criteria operationalize the same principle, but differ in what they protect against. PUP guards against recommendations that are absolutely harmful, while RM guards against recommendations that are relatively inferior.

Testing these criteria empirically poses a fundamental challenge. In real-world settings, we cannot observe the counterfactual (i.e., what the user would have experienced had the recommender not intervened). Nor can we randomly withhold recommendations at scale without ethical and commercial concerns. We therefore turn to simulation. We develop an agent-based model (ABM) grounded in the MTTC framework \cite{vanwee_meta-theory_2025}, in which a heterogeneous population of synthetic travelers interacts with a stack of RS over multiple days. The agents are constructed with psychologically rich profiles including Big Five personality traits, latent variables for trust and autonomy, heterogeneous decision-making paradigms, and dynamic willingness to follow algorithmic suggestions. The RS stack represent the heuristic ranking paradigms that dominate real-world activity platforms and serve as baselines against which we test welfare-oriented alternatives. This setup allows us to run controlled experiments that vary the welfare criterion, observe complete counterfactuals, and disaggregate outcomes by agent type.


We use the above framework to answer three linked research questions:
\begin{enumerate}
    \item[\textbf{RQ1.}] How does welfare-oriented recommendation filtering (using PUP and RM) impact aggregate systemic utility and the distribution of user welfare compared to unconstrained heuristic baselines?
    \item[\textbf{RQ2.}] For what types of users would welfare-oriented recommender systems provide more benefit?
    \item[\textbf{RQ3.}] What is the individual cost of over-recommendation, and to what extent can welfare-oriented filtering reduce the frequency and severity of algorithmic harm compared to making choices organically?
\end{enumerate}

We make three contributions in answering these questions. First, we develop a formal welfare-oriented framework for recommender systems built on two tractable criteria (PUP and RM) that give the system a principled basis for deciding when to recommend and when to abstain. Second, we construct an agent-based simulation testbed grounded in the MTTC that operationalizes heterogeneous travelers interacting with multiple recommender systems in a spatial environment with realistic travel costs and feedback loops. Third, our experimental results identify the conditions under which welfare-oriented recommendation outperforms conventional baselines, characterize which types of users benefit most, and quantify the aggregate cost of over-recommendation. Together, these contributions offer both a normative argument and an empirical toolkit for designing recommender systems that treat user welfare as a primary objective (rather than a byproduct of heuristic ranking).

\begin{figure}
    \centering
    \includegraphics[width=0.75\textwidth]{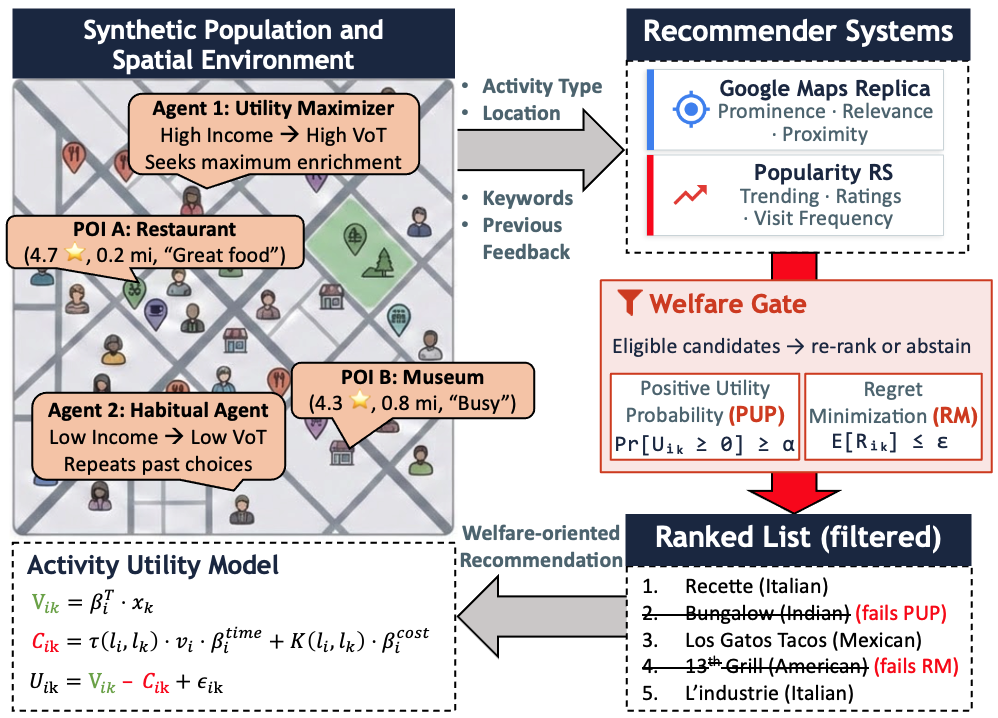}
    \caption{Welfare-oriented filtering framework. (left) Agent-based simulation and the underlying activity utility model; (right) recommendation generation, filtering, and feedback pipeline.}
    \label{fig:overall_framework}
\end{figure}

\section{Related Work}
\label{sec:relatedWork}

\subsection{RS evaluation beyond accuracy}

The dominant evaluation paradigm for recommender systems relies on accuracy metrics derived from revealed preferences: click-through rates, precision at $k$, NDCG, and similar measures that treat user behavior as the ground truth for user satisfaction \cite{pu_user-centric_2011}. A growing body of work argues that this paradigm is insufficient. Pu et al. \cite{pu_user-centric_2011} propose a user-centric evaluation framework that captures dimensions such as perceived recommendation quality, interface adequacy, and decision confidence, none of which are observable from clicks alone. Dokoupil and Peska \cite{dokoupil_how_2025} show that users' perceptions of RS objectives diverge substantially from the objectives that systems actually optimize, suggesting a misalignment between design intent and user experience.

A parallel stream examines RS from a multi-stakeholder perspective. Cruz et al. \cite{cruz_capri-fair_2024} develop fairness-aware recommendation that balances provider exposure against consumer relevance, while Wang et al. \cite{wang_recommending_2025} model multi-sided marketplace objectives hierarchically to prevent recommendation quality from being sacrificed to platform revenue goals. At a broader level, Wagner et al. \cite{wagner_measuring_2021} argue that algorithmic systems reshape the phenomena they claim to measure, creating feedback loops in which standard metrics become performative. Our work aligns with these perspectives using a normative position. Rather than proposing alternative accuracy metrics or fairness constraints, we formalize the conditions under which recommendations should be more conservative or not be issued at all. 


\subsection{Welfare and regret}

The question of whether RS should optimize for what users \emph{choose} or what makes them \emph{better off} has deep roots in welfare economics and is receiving renewed attention in the RS community. Donnelly et al.\ \cite{donnelly_welfare_2024} use a large-scale field experiment to estimate the consumer surplus generated by personalized rankings, finding that personalization generates substantial welfare gains relative to uniform bestseller-based orderings but that the distribution of these gains is uneven. Anwar and Hardt \cite{anwar_recommendation_2025} formalize the distinction between what users are tempted by and what serves their long-term interests, showing that RS can exploit the gap between temptation and welfare by steering users toward high-engagement but low-satisfaction options. Li et al.\ \cite{li_right_2024} study a related tension between repeat consumption and exploration, finding that the ``right'' recommendation depends on whether the system privileges short-term engagement or long-term user satisfaction.

The concept of \emph{regret} in our framework differs from its typical use in the RS literature. In bandit-based RS, regret measures the cumulative loss from suboptimal exploration-exploitation trade-offs over a sequence of decisions \cite{chen_values_2021}. Our RM criterion instead focuses on the disutility experienced when a chosen option turns out worse than a foregone alternative \cite{kahneman_prospect_1979}. This formulation is agent-centric rather than platform-centric, asking what the user lost by following the recommendation rather than what the platform lost by recommending suboptimally. Only a few studies take this approach. For instance, Coppens et al. \cite{coppens_balancing_2024} study when recommending a new activity versus reinforcing a habit is more beneficial. Our PUP and RM criteria provide a principled basis for estimating this disutility (as opposed to relying on confidence thresholds or randomization).


\subsection{Context-aware and location-based recommendation}

Point-of-interest (POI) recommendation has grown into a substantial subfield, driven by the availability of check-in data and the practical importance of location-based services. Recent surveys document the breadth of approaches, spanning collaborative filtering on check-in matrices, sequence-aware models that capture temporal visiting patterns, and graph neural networks that encode spatial relationships \cite{sanchez_context_2025}. Dietz et al.\ \cite{dietz_exploratory_2024} extend POI recommendation to health-promoting activities in urban parks, explicitly modeling the trade-off between activity relevance and spatial accessibility. S{\'a}nchez and Bellog{\'\i}n \cite{sanchez2023bias} identify geographical distance polarization as one dimension of bias in location-based RS, showing that standard ranking algorithms generally favor popular venues which are nearby. 

What distinguishes activity recommendation from conventional item recommendation is the cost of consumption. When a system recommends a product, the primary cost is the purchase price (which can often be recouped via returns); when it recommends a place or activity, the user also incurs travel time, monetary travel costs, and opportunity costs from foregone alternatives. This distinction is well understood in the travel behavior literature, where generalized cost models have been standard since the 1960s \cite{quarmby1967choice}, but it is largely absent from RS research. Gramsch-Calvo et al. \cite{gramsch-calvo_going_2025, gramsch-calvo_importance_2025} study leisure destination choice and show that willingness to travel varies substantially with the social context of the activity and the homophilic preferences of the decision-maker. Ji et al. \cite{ji_meet_2026} disentangle spatial, temporal, and social factors in joint leisure destination choice, confirming that travel cost is a first-order determinant of satisfaction that interacts with activity type.

The feedback loop between recommendations and urban mobility is another distinguishing feature. Mauro et al. \cite{mauro_urban_2025} develop a simulation framework showing that next-venue recommendations increase individual venue diversity but amplify collective inequality by concentrating visits on popular places. Perdomo et al. \cite{perdomo_performative_2020} formalize this as \emph{performative prediction}: when predictions shape the distribution they aim to model, standard retraining can fail to converge to a welfare-optimal equilibrium. Our ABM captures this dynamic by running feedback loops over multiple simulated days, allowing us to study how utility-based abstention or substitution interrupts the concentration effect.


\section{Utility-Based Recommendation Framework}
\label{sec:methods}

This section presents the formal framework that underpins our approach. We begin with the utility model that defines what it means for a recommendation to benefit a user (\S\ref{ssec:utility_model}), then introduce two criteria that operationalize this definition as recommendation policy constraints (\S\ref{ssec:welfare_criteria}), and finally describe how a recommender system can estimate the quantities these criteria require (\S\ref{ssec:rs_architecture}).

\subsection{Activity Utility Model}
\label{ssec:utility_model}

Let $\mathcal{A} = \{a_1, \ldots, a_K\}$ be a finite set of available activities, where each activity $a_k$ is characterized by an attribute vector $\mathbf{x}_k$ describing its intrinsic properties (type, quality, popularity, etc.) and a spatial location $\ell_k$. Agent $i$ is located at position $\ell_i$ and has a latent preference vector $\boldsymbol{\beta}_i$ drawn from a population distribution $P(\boldsymbol{\beta})$.

Our formulation is rooted in random utility theory \cite{mcfadden1973conditional} and the activity-based travel behavior literature \cite{bowman2001activity}. Specifically, we decompose the utility that agent $i$ derives from participating in activity $a_k$ into two components:
\begin{equation}
\label{eq:utility}
    U_{ik} = V_{ik} - C_{ik} + \epsilon_{ik}
\end{equation}
where $V_{ik}$ is the \emph{activity benefit} (the experienced value of the activity itself), $C_{ik}$ is the \emph{travel cost} (the generalized cost of reaching and returning from the activity location), and $\epsilon_{ik}$ is a random utility term capturing unobserved idiosyncratic factors. 

The activity benefit is defined as:


\begin{equation}
\label{eq:activity_benefit}
    V_{ik} = \boldsymbol{\beta}_{i}^{\top} \mathbf{x}_k
\end{equation}
where $\boldsymbol{\beta}_{i}^{\top} \mathbf{x}_k$ captures the match between agent $i$'s preferences and the activity's observable attributes (type, quality, popularity, interest match). All travel in this formulation is treated as derived demand: agents travel solely to reach the destination, and the trip itself carries no intrinsic, escapist, or positional value (for more discussion on travel motivations, we refer the reader to \cite{vanwee_meta-theory_2025}).

Following \cite{ben1985discrete}, the travel cost is defined as:
\begin{equation}
\label{eq:travel_cost}
    C_{ik} = \tau(\ell_i, \ell_k) \cdot v_i \cdot \beta_i^{\text{time}} + \kappa(\ell_i, \ell_k) \cdot \beta_i^{\text{cost}}
\end{equation}
where $\tau(\ell_i, \ell_k)$ is the expected travel time between the agent's current location and the activity, $v_i$ is the agent's value of time, $\kappa(\ell_i, \ell_k)$ is the monetary cost of the trip (fuel and parking), and $\beta_i^{\text{time}}$ and $\beta_i^{\text{cost}}$ are the agent's sensitivity weights for time and cost respectively. All agents travel by car, so travel time depends on prevailing road congestion, making $C_{ik}$ endogenous to aggregate behavior in the simulation (see \S\ref{sec:simulation}). Heterogeneity in how burdensome travel feels is captured through variation in $v_i$ (which depends on income) and the preference weights (which depend on psychological traits such as budget sensitivity and planning orientation).

This two-component structure is central to the paper's argument. A recommendation can match an agent's preferences well ($V_{ik}$ is high) yet still produce negative net utility if the activity is far away and the travel cost outweighs the benefit. Standard RS optimize for engagement signals that correlate with $V_{ik}$ but are largely uninformed about $C_{ik}$. The welfare criteria we propose next address precisely this gap.

\subsection{Utility-Based Criteria}
\label{ssec:welfare_criteria}

We define two criteria that a recommender can evaluate before issuing a recommendation. Both criteria use the same underlying utility model but differ in the question they ask.

\subsubsection{Positive Utility Probability (PUP)}

PUP asks whether the recommended activity is likely to leave the user better off than not participating at all. Formally, a recommendation $r(i) = a_k$ satisfies the PUP criterion at threshold $\alpha$ if:
\begin{equation}
\label{eq:pup}
    \Pr\!\left[U_{ik} \geq 0 \;\middle|\; \hat{\boldsymbol{\beta}}_i\right] \geq \alpha
\end{equation}
where $\hat{\boldsymbol{\beta}}_i$ is the recommender's estimate of agent $i$'s preference vector, and the probability is taken over the residual uncertainty in $\boldsymbol{\beta}_i$ given $\hat{\boldsymbol{\beta}}_i$ and over the random utility term $\epsilon_{ik}$.

Substituting the two-component decomposition, the condition becomes:
\begin{equation}
\label{eq:pup_decomposed}
    \Pr\!\left[V_{ik} + \epsilon_{ik} \geq C_{ik} \;\middle|\; \hat{\boldsymbol{\beta}}_i\right] \geq \alpha
\end{equation}

This formulation makes the role of travel cost transparent: PUP is harder to satisfy for distant activities, for agents with high values of time, and for activities where the preference match is uncertain. Setting $\alpha = 0.5$ recovers a recommendation policy that issues any suggestion expected to produce non-negative utility on average; higher values of $\alpha$ impose progressively stricter confidence requirements.

When the PUP criterion is not met, the recommender can either abstain entirely (letting the agent make an organic choice) or substitute a more conservative recommendation---for example, a closer alternative with a higher probability of non-negative utility.

\subsubsection{Regret Minimization (RM)}

RM asks a complimentary question: is the recommended activity close to the best the user would have found on their own? Define the regret of recommending activity $a_k$ to agent $i$ as:
\begin{equation}
\label{eq:regret}
    R_{ik} = \max_{j \in \mathcal{A}} \left[U_{ij}\right] - U_{ik}
\end{equation}
where $\max_{j \in \mathcal{A}} [U_{ij}]$ is the utility agent $i$ would have obtained from their best available activity had they chosen optimally over the full set $\mathcal{A}$. Regret is always non-negative: $R_{ik} = 0$ when the recommendation coincides with the agent's true optimum, and $R_{ik} > 0$ when a better option exists.

A recommendation satisfies the RM criterion at tolerance $\epsilon$ if:
\begin{equation}
\label{eq:rm}
    \mathbb{E}\!\left[R_{ik} \;\middle|\; \hat{\boldsymbol{\beta}}_i\right] \leq \epsilon
\end{equation}
where the expectation is taken over the joint distribution of $(\epsilon_{ij})_{j \in \mathcal{A}}$ and the residual uncertainty in $\boldsymbol{\beta}_i$ given $\hat{\boldsymbol{\beta}}_i$.

PUP and RM protect against different failure modes. PUP guards against recommendations that are absolutely harmful (negative net utility). RM guards against recommendations that are relatively inferior (positive utility, but less than what the agent would have achieved independently). A recommendation can pass PUP but fail RM if the agent benefits from the activity, but would have benefited more from a different one. Conversely, a recommendation can fail PUP but pass RM in degenerate cases where all available options produce negative utility.

\subsubsection{Decision rules}

Both criteria define a family of recommendation policies indexed by their respective thresholds. Let $\hat{V}_{ik}(\hat{\boldsymbol{\beta}})$ denote the recommender's engagement-based score for activity $a_k$ given its noisy estimate of agent $i$'s preferences. Under a standard (unconstrained) RS, the recommendation is:
\begin{equation}
\label{eq:standard_rs}
    r^{\text{std}}(i) = \arg\max_{k \in \mathcal{A}} \;  \hat{V}_{ik}(\hat{\boldsymbol{\beta}})
\end{equation}

Under PUP-constrained recommendation:
\begin{equation}
\label{eq:pup_policy}
    r^{\text{PUP}}(i) = \begin{cases}
        \arg\max\limits_{k \in \mathcal{A}_\alpha} \hat{V}_{ik}(\hat{\boldsymbol{\beta}}) & \text{if } \mathcal{A}_\alpha \neq \emptyset \\[4pt]
        \varnothing & \text{otherwise}
    \end{cases}
\end{equation}
where $\mathcal{A}_\alpha = \{a_k \in \mathcal{A} : \Pr[U_{ik} \geq 0 \mid \hat{\boldsymbol{\beta}}_i] \geq \alpha\}$ is the set of activities that pass the PUP threshold. The RS selects the highest-scoring activity from the admissible set; if no activity qualifies, it abstains.


Under RM-constrained recommendation:
\begin{equation}
\label{eq:rm_policy}
    r^{\text{RM}}(i) = \begin{cases}
        \arg\max\limits_{k \in \mathcal{A}_\epsilon} \hat{V}_{ik} (\hat{\boldsymbol{\beta}}) & \text{if } \mathcal{A}_\epsilon \neq \emptyset \\[4pt]
        \varnothing & \text{otherwise}
    \end{cases}
\end{equation}
where $\mathcal{A}_\epsilon = \{a_k \in \mathcal{A} : \mathbb{E}[R_{ik} \mid \hat{\boldsymbol{\beta}}_i] \leq \epsilon\}$. For both RM- and PUP-constrained recommendations, if at least one recommendation survives the filtering for the chosen activity subtype, the model does not abstain.

To formalize the trade-off inherent in welfare gating, let $W_{ik}$ represent a generalized scalar welfare score for activity $a_k$ recommended to agent $i$. The recommender system applies a confidence threshold $\theta \in [\theta_{\min}, \theta_{\max}]$, defining the admissible set of recommendations as $\mathcal{A}_\theta = \{a_k \in \mathcal{A} : W_{ik} \ge \theta\}$. Both PUP and RM are special cases of this formulation. As $\theta$ increases, the admissible set $\mathcal{A}_\theta$ monotonically shrinks. We can now show the following:

\begin{proposition}
Let $p(\theta) = \Pr[\mathcal{A}_\theta \neq \varnothing]$ denote the probability that at least one candidate survives the welfare filter at strictness threshold $\theta$. Let $\mu(\theta) = \mathbb{E}[U_{ik^*(\theta)} \mid \mathcal{A}_\theta \neq \varnothing]$ denote the expected net utility of the issued recommendation, conditional on the RS not abstaining. Let $\bar{U}^{\text{org}}$ represent the mean net utility under organic choice. The unconditional mean net utility under a gating policy at threshold $\theta$ is:
$$\bar{U}(\theta) = p(\theta) \mu(\theta) + (1 - p(\theta)) \bar{U}^{\text{org}}$$

Assume the unconstrained baseline RS provides value over organic choice ($\mu(\theta_{\min}) > \bar{U}^{\text{org}}$) and the residual utility uncertainty is normally distributed. For any filtering criterion where the admissible set is monotonically decreasing in $\theta$, $\bar{U}(\theta)$ is strictly quasi-concave on $[\theta_{\min}, \theta_{\max}]$, and there exists a unique interior optimal threshold $\theta^* \in (\theta_{\min}, \theta_{\max})$ that maximizes expected user welfare.  
\end{proposition}

\begin{proof}
See Appendix~\ref{app:proof}.
\end{proof}

These policies have two effects. First, \emph{filtering}: activities with high engagement scores but poor expected welfare are removed from consideration. Second, \emph{re-ranking}: among the surviving candidates, the standard engagement-based ordering is preserved, but the effective top recommendation may change because formerly top-ranked options have been filtered out. This means the system can shift toward more conservative recommendations (e.g., closer, better-matched activities) rather than abstaining entirely.

\subsection{Recommender Architecture}
\label{ssec:rs_architecture}

To evaluate PUP and RM, the recommender must estimate $V_{ik}$ and $C_{ik}$ for each candidate activity. We describe how each component is estimated and how the estimates are combined into the welfare criteria.

\subsubsection{Estimating activity benefit}

The RS maintains a personalization model for each user, updated from post-activity feedback. The estimated activity benefit for a candidate activity $a_k$ is constructed as a weighted sum of the platform's \emph{base score} $B_k$ and the \emph{personalization score} $P_{ik}$:
\begin{equation}
\label{eq:v_hat}
    \hat{V}_{ik} = (1 - \lambda) \cdot B_k + \lambda \cdot P_{ik}
\end{equation}
where $\lambda \in [0, 1]$ is a personalization weight. The base score $B_k$ captures aggregate quality signals (e.g., ratings, review volume, proximity), while the personalization score $P_{ik}$ is derived from the accumulated affinities. We distinguish two types of user-level affinity scores: a \emph{place affinity} $\phi_{ik}^{\text{place}}$ for previously visited places, and a \emph{keyword affinity} $\phi_{iw}^{\text{kw}}$ for content descriptors (e.g., ``outdoor,'' ``Italian,'' ``live music''):
\begin{equation}
\label{eq:personalization}
    P_{ik} = \tfrac{1}{2}\!\left(\Omega \cdot \phi_{ik}^{\text{place}} + (1-\Omega) \cdot \bar{\phi}_{i,\text{kw}(k)} + 1\right),
\end{equation}
where $\bar{\phi}_{i,\text{kw}(k)}$ is the mean keyword affinity across the keywords associated with place $k$ and the parameter $\Omega \in [0,1]$ controls the relative weight on place-level versus keyword-level affinities. The specific platform instantiations of $B_k$ as well as values of $\Omega$ used in our simulation are described in Appendix~\ref{app:rs_stack}.

After user $i$ completes activity $a_k$, the system receives a feedback signal $f_{ik}$, and affinities are updated incrementally:
\begin{equation}
\label{eq:affinity_update}
    \phi_{ik}^{(t+1)} = \phi_{ik}^{(t)} + \gamma \cdot f_{ik}
\end{equation}
where $\gamma$ is a learning rate. The feedback signal $f_{ik}$ may range from a simple binary like/dislike to a continuous satisfaction measure; in our instantiation (\S\ref{ssec:feedback}), both channels are used.

The RS also estimates a noisy measure for travel cost via Eq.~\ref{eq:travel_cost}, where $\hat{v}_i$, $\hat{\beta}_i^{\text{time}}$, and $\hat{\beta}_i^{\text{cost}}$ are learned from the user's feedback history. Conceptually, users who consistently report dissatisfaction with distant activities despite high preference match are inferred to have higher time sensitivity. The system maintains a per-user uncertainty estimate $\sigma_i$ that decreases with the number of feedback observations, which feeds into the PUP computation described in Appendix~\ref{app:computing_PUP}. Implementation details of the learning mechanism (asymmetric update rates and bounds) are provided in Appendix~\ref{app:tce}.




\section{Simulation Model}
\label{sec:simulation}

To evaluate the welfare-oriented criteria introduced in \S\ref{sec:methods}, we require a testbed in which (i) agents make activity and travel decisions under heterogeneous behavioral assumptions, (ii) recommender systems intervene with suggestions that carry real travel costs, (iii) what each agent would have experienced absent the recommendation (the counterfactual) is observable, and (iv) the feedback loop between agents and recommenders unfolds over multiple decision episodes. No public observational dataset satisfies all four requirements simultaneously. We therefore construct an agent-based model (ABM) grounded in the Meta-Theory for Travel Choices (MTTC) \cite{vanwee_meta-theory_2025}, in which a heterogeneous population of synthetic travelers interacts with a set of recommender systems within a spatial urban environment. This section describes each component.

\subsection{Agent specification}
\label{ssec:agents}

Each agent represents a traveler characterized by personal characteristics, a decision-making paradigm, and contextual factors, following the MTTC building blocks.

\subsubsection{Personal characteristics and behavioral heterogeneity}

Agents are endowed with socio-demographic attributes (age, income) that shape their preferences. Income determines value of time (VOT, $v_i$) as a fraction of the agent's hourly wage rate, $v_i = \gamma_w \cdot (\text{income}_i / H)$, where $H$ is the annual work hours and $\gamma_w \in (0,1]$ is a scaling factor (see Appendix~\ref{app:value_of_time} for calibration). Mode choice is fixed (i.e., all agents travel by car) to isolate the effect of recommendation policy on welfare outcomes.

Beyond demographics, each agent carries a Big Five personality vector (openness, conscientiousness, extraversion, agreeableness, neuroticism) and four latent psychological variables: \emph{maximization tendency} (the disposition to search exhaustively for the best option rather than settle for a satisfactory one), \emph{trust in platforms} (baseline willingness to rely on algorithmic recommendations), \emph{autonomy preference} (desire to make independent decisions), and \emph{algorithmic awareness} (understanding of how RS operate). These traits shape behavioral parameters such as risk aversion, budget sensitivity, and the dynamic willingness to accept recommendations ($\eta_i$, \S\ref{ssec:rec_interaction}). This design is empirically motivated by Shi et al. \cite{shi_antecedents_2021}, who show that travelers process RS recommendations through both cognitive and emotional trust channels. Our trust-in-platforms and autonomy-preference variables capture this heterogeneity. The full mapping from personality traits to behavioral coefficients is provided in Appendix~\ref{app:psychological_traits}.

\subsubsection{Decision-making paradigms}

Agents differ in how they evaluate and select alternatives, following the MTTC's recognition that not all travelers are utility maximizers. The model implements five paradigms: \emph{utility maximizers} select the option with the highest expected utility \cite{mcfadden1973conditional}; \emph{regret minimizers} select the option with the lowest anticipated regret relative to foregone alternatives \cite{loomes1982regret}; \emph{prospect-theory agents} overweight potential losses relative to gains \cite{kahneman_prospect_1979}; \emph{satisficers} accept the first option that exceeds a personal acceptability threshold \cite{caplin2011search}; and \emph{habitual agents} repeat their most recent choice unless it becomes unavailable \cite{gonzalez2008understanding}.


\subsection{Spatial environment}
\label{ssec:city}

Agents operate in a synthetic grid-based city with four zone types (residential, employment, leisure, mixed), each populated with synthetically generated points of interest (POIs) drawn from category distributions appropriate to the zone type. Road congestion is endogenous: traffic volumes at each time step feed into a Bureau of Public Roads (BPR) volume-delay function that increases travel times, making travel cost $C_{ik}$ responsive to aggregate behavior. Full details of the spatial model, POI generation, and congestion mechanics are provided in Appendix~\ref{app:spatial}.

\subsection{Daily activity scheduling}
\label{ssec:scheduling}

Agents follow a daily schedule that begins with a home-to-work commute trip, followed by a leisure activity after work. The choice of leisurely activity follows a two-stage scheduling process. First, agents choose a  activity subtype to partake in. Then, agents interact with RS to get recommendations for that subtype, which they either choose to accept or disregard, instead making an organic choice. 



\subsubsection{Activity type and location selection}

After work, each agent evaluates the net utility of seven leisure activity types: dine-in food, takeout food, live music, fitness, caf\'{e} meetups, museums, and parks. For each segment~$s$, the agent computes $U_{is}$ following the two-component utility model of Eq.~\ref{eq:utility}, where $V_{is}$ incorporates base utility, personality adjustments, and interest match, and $C_{is}$ captures the expected travel disutility by car to the candidate location (which varies across agents through their time and cost preference weights, as specified in Eq.~\ref{eq:travel_cost}). The agent selects among segments via a softmax over net utilities.

\subsubsection{Recommendation interaction}
\label{ssec:rec_interaction}

For each leisure segment, the recommender stack produces a ranked list of candidate places. The agent observes the top recommendation and evaluates a dynamic \emph{willingness-to-accept} parameter $\eta_i \in [0, 1]$, which determines the probability of following the recommendation. We describe the full parameterization in Appendix~\ref{app:eta}.

If the agent accepts the recommendation (with probability $\eta_i$), the leisure activity's location is set to the recommended place and the recommendation source is recorded. If the agent declines, they proceed with an organic location choice. This mechanism means the RS influences but does not determine the agent's behavior, consistent with empirical evidence that recommendation acceptance is partial and context-dependent \cite{ninomiya_determinants_2025}.

\subsection{Feedback and learning}
\label{ssec:feedback}

After completing a leisure activity, agents generate feedback on accepted recommendations. A composite feedback signal combines the experienced activity utility and travel utility, adjusted by preference fit and the agent's feedback sensitivity. This signal is transmitted through two channels. First, it passes through a sigmoid function to produce a stochastic binary like/dislike outcome (with feedback strength proportional to the agent's sensitivity), which updates the originating RS platform's personalization model via the affinity learning rule in Eq.~\ref{eq:affinity_update}. Second, a continuous feedback record is generated containing the satisfaction score $s_{ik} \in [-1, 1]$, the realized activity utility, travel utility, and generalized travel cost. This continuous signal updates the per-user VOT parameters, enabling the welfare layer to refine its estimates of $\hat{V}_{ik}$ and $\hat{\beta}_i^{\text{time}}$ over time.

The welfare estimation quality improves over time as the RS accumulates feedback. Early in a user's interaction history, the estimates $\hat{V}_{ik}$ and $\hat{\beta}_i^{\text{time}}$ are uncertain, leading to wider confidence intervals and more frequent abstention under strict PUP or RM thresholds. As feedback accumulates, the estimates tighten and the RS can recommend with greater confidence. This creates a natural exploration--exploitation dynamic: the system is conservative when uncertain and becomes more assertive as it learns.  

\section{Experimental Design and Results}
\label{sec:results}

We evaluate the welfare-oriented RS framework through three analyses. First, we compare aggregate welfare metrics across recommendation conditions and trace the PUP and RM threshold frontiers (\S\ref{ssec:aggregate}). Second, we disaggregate outcomes by decision paradigms to assess distributional effects (\S\ref{ssec:heterogeneity}). Third, we quantify the cost of over-recommendation (\S\ref{ssec:orc}). Unless otherwise noted, experiments use 250 agents on a 18$\times$18 city grid over 60 simulated days, averaged across 5 random seeds. The first 59 days are used to train the RS for personalization through feedback, while the last day is used to assess the performance of each recommendation policy. Each treatment arm is run with the same set of seeds, ensuring that any difference in outcomes is attributable to the recommendation policy rather than to stochastic variation in agent characteristics or scheduling. The simulation code can be found \href{https://github.com/ekinugurel/welfare_oriented_rec_sys/blob/main/paper_figures_and_tables.ipynb}{here}.

We compare five recommendation policies:

\begin{enumerate}
    \item \textbf{No RS}: Agents select leisure segments and locations entirely from their own preferences and spatial proximity. This condition establishes the counterfactual: what agents would experience without algorithmic intervention.

    \item \textbf{Standard RS}: The set of RS operates as described in \S\ref{ssec:rs_architecture} and Appendix~\ref{app:rs_stack}, ranking candidates by engagement signals (prominence, relevance, proximity, popularity, personalization) without welfare constraints.

    \item \textbf{PUP-constrained RS}: The recommender applies the Positive Utility Probability filter (\S\ref{ssec:welfare_criteria}) before issuing recommendations. Only candidates satisfying $\Pr[U_{ik} \geq 0 \mid \hat{\boldsymbol{\beta}}_i] \geq \alpha$ are eligible; the surviving candidates are then ranked by the same heuristics Standard RS uses. 

    \item \textbf{RM-constrained RS}: The recommender applies the Regret Minimization filter. Only candidates satisfying $\mathbb{E}[R_{ik} \mid \hat{\boldsymbol{\beta}}_i] \leq \epsilon$ are eligible; the surviving candidates are then ranked by the same heuristics Standard RS uses. The threshold $\epsilon$ is calibrated from the empirical distribution of expected regret under the Standard RS condition.

    \item \textbf{Oracle RS}: A hypothetical recommender with perfect knowledge of each agent's true preference vector $\boldsymbol{\beta}_i$, including their value of time, time and cost sensitivities, leisure preferences, and subtype-specific utility bonuses. It computes the exact net utility $U_{ik} = V_{ik} - C_{ik}$ for every candidate and returns the optimal ranking. This condition provides an upper bound on achievable welfare.
\end{enumerate}

\subsection{Aggregate welfare comparison and threshold analysis}
\label{ssec:aggregate}

We first evaluate the welfare-oriented framework from a systemic perspective, analyzing whether PUP and RM constraints lift the overall performance and equity of the recommendation ecosystem. Table~\ref{tab:aggregate_results} reports aggregate welfare metrics across recommendation conditions. The existence of an activity RS improves mean leisure-trip utility by $0.04$ but barely reduces the negative trip-rate (fraction of trips with $U < 0$), indicating that many recommendations incur travel costs that still offset activity gains. Among welfare-oriented conditions, PUP-0.6 and RM ($\epsilon{=}0.94$) achieve the highest mean utility ($\bar{U} = 0.231$ and $0.242$, respectively) while abstaining on 45-62\% of recommendations. Both conditions also reduce the negative-trip rate to $0.228-0.248$ and improve upon the Standard RS in Gini coefficient, suggesting that selective recommendation improves welfare without concentrating gains among a narrow subset of agents. The Oracle upper bound ($\bar{U} = 0.377$) confirms substantial headroom. Welfare-oriented recommendations close roughly two-thirds of the gap between Standard RS and Oracle. 

\begin{table}[!]
\centering
\caption{Aggregate welfare metrics by recommendation condition. $\bar{U}$: mean net trip utility; $\sigma_U$: standard deviation across seeds; Neg.\ rate: fraction of trips with $U < 0$; Abst.: fraction of recommendation opportunities where the RS withheld a suggestion; Gini: inequality in the trip-utility distribution.}
\label{tab:aggregate_results}
\small
\begin{tabular}{@{}lccccc@{}}
\toprule
\textbf{Condition} & $\bar{U}$ & $\sigma_U$ & \textbf{Neg.\ rate} & \textbf{Abst.} & \textbf{Gini} \\
\midrule
No RS              & $0.052$ & $0.016$ & $0.415$ & $1.000$ & $0.159$ \\
Standard RS        & $0.089$ & $0.053$ & $0.367$ & $0.332$ & $0.151$ \\
PUP-0.3           & $0.214$  & $0.041$ & $0.256$ & $0.563$ & $0.155$ \\
\textbf{PUP-0.6}            & $\textbf{0.231}$  & $\textbf{0.035}$ & $\textbf{0.248}$ & $\textbf{0.620}$ & $\textbf{0.144}$ \\
\textbf{RM} ($\epsilon{=}\textbf{0.94}$) & $\textbf{0.242}$  & $\textbf{0.035}$ & $\textbf{0.228}$ & $\textbf{0.454}$ & $\textbf{0.142}$ \\
\bottomrule
Oracle & $0.377$ & $0.023$ & $0.176$ & $0.409$ & $0.159$\\
\bottomrule
\end{tabular}
\end{table}

Next, we vary the key parameters in PUP- and RM-constrained models to identify patterns and the best-performing variants. Specifically, we trace the PUP frontier by plotting mean net utility and abstention rate against $\alpha$ (Figure~\ref{fig:pup_frontier}), and similarly for RM against $\epsilon$ (Figure~\ref{fig:rm_frontier}). The PUP frontier reveals two optimal $\alpha$ values. Mean utility improves as $\alpha$ increases from 0 to 0.3, then stagnates until another maximum at $\alpha = 0.6$, and then declines as the RS increasingly abstains. The region near $\alpha = 0.3$ marks that of best welfare improvements without excessive abstention. Beyond $\alpha = 0.7$, the RS abstains on the majority of opportunities, and at $\alpha \geq 0.9$ it withholds all recommendations.

The RM frontier (Figure~\ref{fig:rm_frontier}) shows that a high regret ceiling ($\epsilon$ corresponding to the 95th percentile of the empirical regret distribution) achieves the highest mean net utility while also maintaining a reasonable abstention rate (around $45\%$). Tighter ceilings (50th--70th percentile) filter out too many possible alternatives, which also reduces mean net utility (as agents revert to their organic choices).

\begin{figure}[!]
\centering
\includegraphics[width=0.75\textwidth]{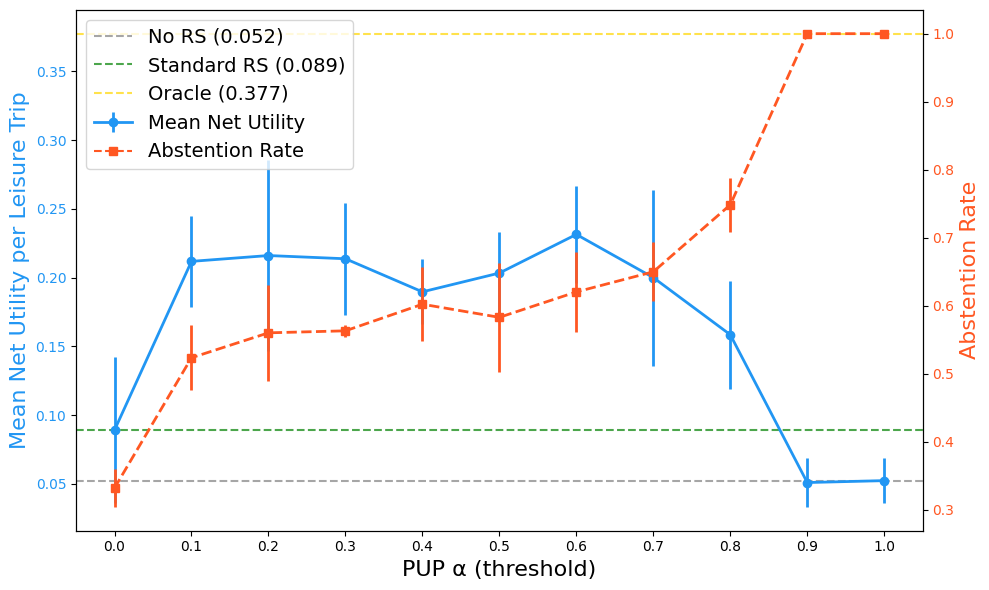}
\caption{PUP threshold frontier. Mean net utility (left axis) and abstention rate (right axis) as $\alpha$ varies. Horizontal lines mark the No RS, Standard RS, and Oracle baselines.}
\label{fig:pup_frontier}
\end{figure}

\begin{figure}[!]
\centering
\includegraphics[width=0.75\textwidth]{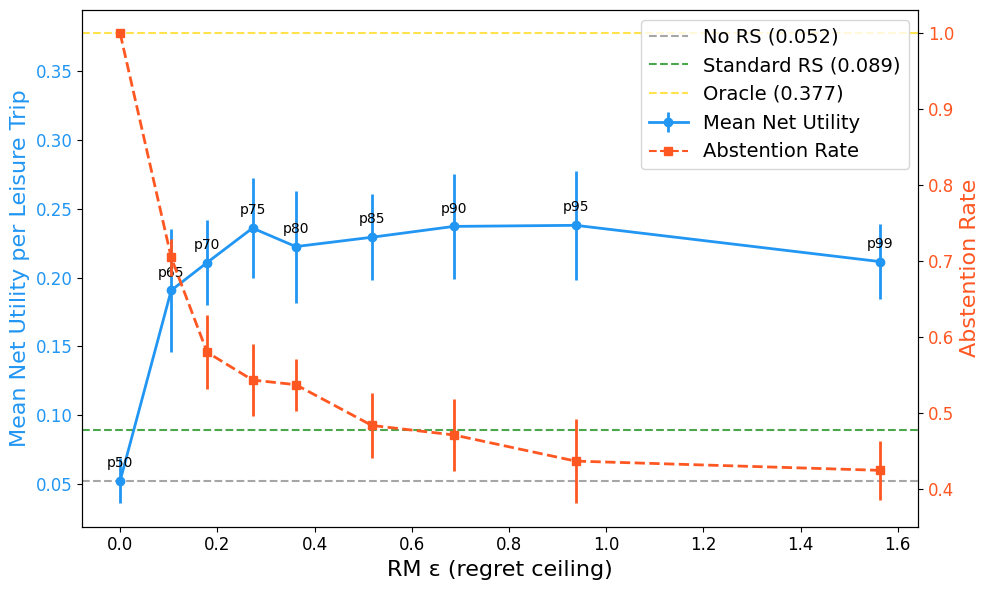}
\caption{RM threshold frontier. Mean net utility and abstention rate as the regret ceiling $\epsilon$ varies across calibration percentiles.}
\label{fig:rm_frontier}
\end{figure}

\subsection{Heterogeneity analysis}
\label{ssec:heterogeneity}

For each condition, we disaggregate mean utility by decision paradigm to test whether welfare-oriented recommendation benefits some behavioral types more than others. Figure \ref{fig:heterogeneity_paradigm} shows the change in mean utility relative to the No RS baseline, broken down by the five decision paradigms. The Oracle produces the largest gains uniformly, with prospect agents ($+0.263$) benefiting most and utility agents ($+0.091$) benefiting least. Under Standard RS, utility agents experience a \emph{decrease} in mean net utility, suggesting that unconstrained recommendations can actively harm agents who apply those heuristics. PUP- and RM-constrained models benefit satisficing and prospect agents the most. These findings support the thesis that, while the magnitude of benefit is sensitive to an agent's underlying decision-making style, the sign of the effect remains positive under welfare-oriented constraints.

\begin{figure}[!]
\centering
\includegraphics[width=0.75\textwidth]{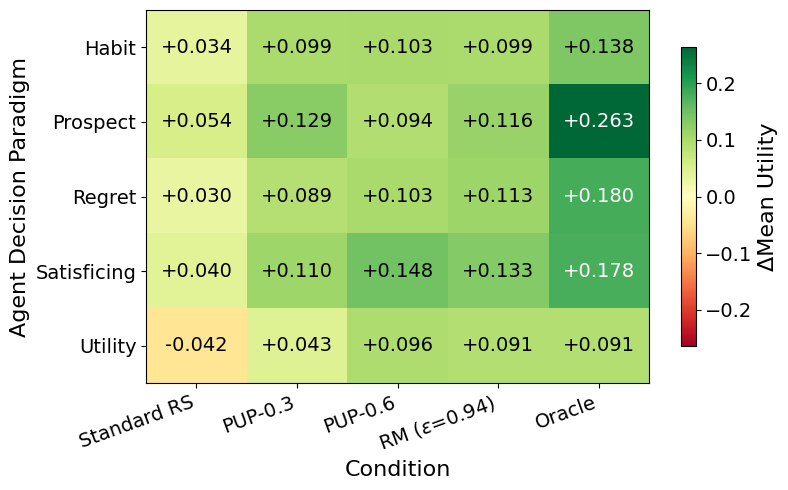}
\caption{Change in mean utility relative to organic baseline, by decision paradigm and condition.}
\label{fig:heterogeneity_paradigm}
\end{figure}



\subsection{Cost of over-recommendation}
\label{ssec:orc}

In \S\ref{ssec:aggregate} we established that welfare-oriented filtering raises aggregate utility. In this section, we turn to the micro-level cost of algorithmic misfires. An RS might look successful on average while still actively disrupting the plans of specific users. We quantify this harm by computing, for each accepted recommendation under the various RS conditions, whether the agent would have been better off under the No RS counterfactual. Specifically, let the \emph{over-recommendation cost} be defined as:
\begin{equation}
\label{eq:over_rec_cost}
    \text{ORC} = \frac{1}{|\mathcal{R}^{-}|} \sum_{(i,k) \in \mathcal{R}^{-}} \left(U_{ik}^{\text{organic}} - U_{ik}^{\text{rec}}\right),
\end{equation}
where $\mathcal{R}^{-}$ is the set of agent-trip pairs where the recommendation produced lower utility than the organic counterfactual, $U_{ik}^{\text{rec}}$ is the realized utility under the recommendation, and $U_{ik}^{\text{organic}}$ is the realized utility in the matched No RS run.

\begin{table}[!]
\centering
\caption{Over-recommendation cost (ORC) by condition.  Harmed \%: fraction of agents whose recommended trip yielded lower utility than their organic trip; Improved \%: fraction whose utility increased; ORC: mean utility loss per harmed agent (Eq.~\ref{eq:over_rec_cost}).}
\label{tab:over_rec}
\small
\begin{tabular}{@{}lccc@{}}
\toprule
\textbf{Condition} & \textbf{Harmed \%} & \textbf{Improved \%} & \textbf{Mean ORC} \\
\midrule
Standard RS             & 39.7\% & 43.8\% & $0.387 \pm 0.016$ \\
\textbf{PUP-0.3}                 & \textbf{31.4\%} & \textbf{50.4\%} & $\textbf{0.383} \pm \textbf{0.017}$ \\
PUP-0.6                 & 32.2\% & 51.0\% & $0.383 \pm 0.024$ \\
\textbf{RM ($\epsilon{=}0.94$)} & \textbf{31.3\%} & \textbf{51.8\%} & $\textbf{0.351} \pm \textbf{0.026}$ \\
\bottomrule
Oracle                  & 27.8\% & 53.4\% & $0.375 \pm 0.013$ \\
\bottomrule
\end{tabular}
\end{table}

Table~\ref{tab:over_rec} reports ORC metrics from the welfare experiment. The Standard RS harms $39.7\%$ of agents (nearly as many as it helps at $43.8\%$) and imposes the highest per-agent ORC ($0.387$). Welfare-oriented gating substantially reduces harm: RM ($\epsilon{=}0.94$) and PUP-0.3 lower the harmed fraction to $\approx 31\%$ and reduce the mean ORC by up to $9.30\%$, reflecting the value of withholding recommendations when expected utility is uncertain. The Oracle harms the fewest percentage of agents at $27.8\%$ while also improving the utility of $53.4\%$, confirming that better preference knowledge translates directly into fewer welfare-reducing recommendations. We note that the Oracle only optimizes place ranking under the utility model, but realized welfare also depends on the organic alternative and travel conditions (i.e., congestion), which the Oracle does not account for (hence why it also harms a fraction of the agents).



\section{Discussion and Conclusion}
\label{sec:discussion}

We have introduced a welfare-oriented framework for evaluating and filtering recommendations in the context of activity-travel behavior. The framework formalizes two utility-based criteria (PUP and RM) that give RS a principled basis for welfare-oriented suggestions. Our experimental results show that welfare-oriented gating reduces the fraction of harmed agents by roughly 30\% relative to the standard heuristic baseline, increases the mean net utility of leisurely trips in our synthetic simulation, and is useful for a variety of user decision-making paradigms in activity-travel behavior.

Our work speaks to a gap between the POI recommendation literature and the behavioral sciences in which these recommendations have systemic effects. Travel behavior research has long recognized that activity and destination choices are governed by heterogeneous preferences, generalized travel costs, and decision heuristics that vary across individuals and contexts~\cite{vanwee_meta-theory_2025}. Yet POI recommender systems are largely developed and evaluated within the information retrieval tradition. For recommender systems research to evolve in domains where its outputs carry real spatiotemporal costs, the field must engage with the behavioral disciplines that study these costs. The framework we propose is one step in that direction.

Several limitations should be acknowledged. First, the agent behavioral model, while grounded in the MTTC and parameterized with empirically motivated distributions, relies on synthetic agents rather than data from real travelers. The preference structures, personality-to-behavior mappings, and willingness-to-accept dynamics are plausible but have not been validated against individual-level behavioral data. Second, the spatial environment is a stylized grid city rather than a real urban network, and all agents travel by car, which limits generalizability to multi-modal settings. Finally, the welfare criteria assume that the recommender can estimate travel costs and preference parameters with reasonable accuracy. In practice, the quality of these estimates depends on data that platforms may not currently collect or that users may be unwilling to share. Our robustness checks in Appendix~\ref{app:robustness} stress-test the parameter choices in our simulation for these kinds of variations and find that, in 34 out of 35 cases, the mean utility ordering presented in our Table \ref{tab:aggregate_results} remains the same. 

The difficulty of studying RS welfare effects in isolation calls for mixed-methods approaches. User surveys and stated-preference experiments can inform the distributions from which synthetic agents are drawn (e.g., calibrating the relationship between personality traits and recommendation acceptance or eliciting willingness-to-travel thresholds for different activity types). Focus groups can surface qualitative dimensions of recommendation harm (frustration, trust erosion) that are difficult to capture in a utility function but that matter for long-term platform engagement. Integrating such empirical inputs into simulation frameworks like ours would strengthen both the internal validity of the model and the external validity of its policy implications. 

\bibliographystyle{plain}
\bibliography{acmart}

\newpage
\appendix
{\Huge\bfseries Appendices}
\section{Proof of Proposition 1}
\label{app:proof}

\begin{proof}
By definition, the non-abstention probability $p(\theta)$ is continuous and strictly decreasing in $\theta$, bounded by $p(\theta_{\min}) = 1$ (unconstrained recommendation) and $p(\theta_{\max}) = 0$ (strict filtering precluding all candidates). Because the filter progressively eliminates the lowest-scoring candidates, the conditional expected utility of a surviving recommendation, $\mu(\theta)$, is non-decreasing in $\theta$.

Differentiating $\bar{U}(\theta)$ with respect to $\theta$ yields:
\begin{equation}
    \bar{U}'(\theta) = p(\theta) \mu'(\theta) + (\mu(\theta) - \bar{U}^{\text{org}}) p'(\theta),
\end{equation}
\noindent where $p(\theta) \mu'(\theta)$ is strictly non-negative and represents the marginal welfare gain from filtering out sub-optimal options and $(\mu(\theta) - \bar{U}^{\text{org}}) p'(\theta)$ represents the marginal penalty of increased abstention. It is non-positive because $p'(\theta) \le 0$ and $\mu(\theta) \ge \mu(\theta_{\min}) > \bar{U}^{\text{org}}$.

Evaluating the derivative at the boundaries reveals a change in sign. At minimum strictness, $p(\theta_{\min}) = 1$ and $p'(\theta_{\min})$ approaches zero, leaving $\bar{U}'(\theta_{\min}) > 0$. The RS strictly improves upon organic choice. As $\theta \to \theta_{\max}$, universal abstention occurs ($p(\theta_{\max}) \to 0$), causing $\bar{U}(\theta)$ to converge to the organic baseline $\bar{U}^{\text{org}}$ from above, implying a negative trajectory.

Because $\bar{U}'(\theta)$ is continuous and changes sign from positive to negative, the Intermediate Value Theorem guarantees the existence of at least one interior critical point $\theta^* \in (\theta_{\min}, \theta_{\max})$ where $\bar{U}'(\theta^*) = 0$.

Finally, because the residual uncertainty $\epsilon_{ik}$ in the utility model is normally distributed, its probability density function is strictly log-concave. This ensures a monotonically increasing hazard rate for the welfare-score distribution, guaranteeing that the derivative $\bar{U}'(\theta)$ crosses zero exactly once. Therefore, the critical point $\theta^*$ is the unique global maximizer of expected user welfare.
\end{proof}

\section{Simulation Details}
\label{app:simulation}

\subsection{Recommender Stack Details}
\label{app:rs_stack}

The simulation implements a multi-platform recommender stack that mirrors the heterogeneous RS landscape users encounter in practice. Two platform types are instantiated, each producing a base score $B_k$ that is then blended with personalization via Eq.~\ref{eq:v_hat}. We check the robustness of all parametric assumptions (detailed below) in \S\ref{app:robustness}.

\subsubsection{Google Maps Replica} 
This platform computes the base score $B_k$ for candidate place $k$ as a convex combination of three normalized signals:
\begin{equation}
    B_k = w_{\text{prom}} \cdot \text{Prom}_k + w_{\text{rel}} \cdot \text{Rel}_k + w_{\text{prox}} \cdot \text{Prox}_k
\end{equation}
\noindent where the weights $w_{\text{prom}} = 0.45$, $w_{\text{rel}} = 0.35$, and $w_{\text{prox}} = 0.20$ are normalized to sum to one. \emph{Prominence} (\text{Prom}) combines the place's average rating (scaled to $[0,1]$) and the log-transformed review count (min-max normalized), weighted 0.55 and 0.45 respectively. \emph{Relevance} (\text{Rel}) is a Jaccard similarity between the place's keyword tags and the user's query keywords (75\% weight) plus interest keywords (25\% weight). \emph{Proximity} (\text{Prox}) applies exponential distance decay: $\text{Prox}_k = e^{-d_{ik}/d_0}$, where $d_{ik}$ is the Euclidean grid distance scaled to kilometers and $d_0 = 5$ km. The final score blends this base with personalization: $\hat{V}_{ik}^{\text{act}} = (1 - \lambda)\, B_k + \lambda\, P_{ik}$, with default $\lambda = 0.10$ (clamped to $[0, 0.35]$).

\subsubsection{Popularity-Based Recommender} 
This platform computes $B_k$ from three signals: average rating (scaled to $[0,1]$, weight 0.25), log-transformed review count (weight 0.35), and a general popularity signal based on total visit count (weight 0.40). Raw scores are min-max normalized before weighting. The default personalization weight is $\lambda = 0.12$ (also clamped to $[0, 0.35]$). This platform does not use proximity or relevance signals; it represents platforms like OpenTable or Ticketmaster where popularity and social proof dominate ranking \cite{kuo2015contextual, nilashi2018travelers}.

\subsubsection{Subtype Routing} 
We route each leisure subtype to the appropriate platform. Table~\ref{tab:rs_routing} shows the mapping. When multiple platforms serve a subtype (e.g., food), recommendations from each platform are returned separately and the agent selects from the merged set. This set of RS are selected from popular apps/companies operating in these domains within the continental United States in 2026.

\begin{table}[h]
\centering
\caption{Leisure subtype to recommender platform routing.}
\label{tab:rs_routing}
\small
\begin{tabular}{@{}ll@{}}
\toprule
\textbf{Leisure subtype} & \textbf{Platform(s)} \\
\midrule
Food (dine-in) & Google Maps replica + Popularity RS (i.e., OpenTable) \\
Food (takeout) & Google Maps replica + Popularity RS (i.e., OpenTable) \\
Live music & Popularity RS (i.e., Spotify) \\
Workout / fitness & Popularity RS (i.e., ClassPass) \\
Caf\'{e} & Google Maps replica \\
Museum & Google Maps replica\\
Park & Google Maps replica \\
\bottomrule
\end{tabular}
\end{table}

\subsubsection{Personalization Learning}
Both platform types share the same personalization mechanism (Eq.~\ref{eq:personalization}), in which we use $\Omega = 0.6$ to give higher weight to place affinity (motivated by recency bias in human mobility \cite{barbosa2015effect}. When a user completes an activity and provides feedback, place-level and keyword-level affinities are updated incrementally (Eq.~\ref{eq:affinity_update}) with a learning rate of $\gamma = 0.12$. Negative feedback applies a stronger update ($-\gamma$) than positive feedback ($+\gamma$), and keyword affinities update at 70\% of the place-level rate. Over the course of the simulation, this causes the personalization score $P_{ik}$ to drift toward places and content types that the agent has consistently enjoyed, and away from those that produced negative experiences.

\subsubsection{Computing PUP and RM}
\label{app:computing_PUP}
Given estimates $\hat{V}_{ik}$ and $\hat{C}_{ik}$, the RS computes the welfare criteria as follows.

For PUP, the RS estimates the probability that net utility is non-negative. Assuming the residual uncertainty $\epsilon_{ik} \sim \mathcal{N}(0, \sigma_i^2)$ with variance $\sigma_i^2$ learned from prediction errors in the feedback history:
\begin{equation}
\label{eq:pup_computation}
    \Pr\!\left[U_{ik} \geq 0\right] \approx \Phi\!\left(\frac{\hat{V}_{ik} - \hat{C}_{ik}}{\sigma_i}\right)
\end{equation}
where $\Phi(\cdot)$ is the standard normal CDF. The recommendation passes PUP at threshold $\alpha$ if this probability exceeds $\alpha$. 

For RM, the RS estimates expected regret by comparing the recommended activity against the set of candidates it would have surfaced in a standard ranking. Let $\hat{U}_{ij} = \hat{V}_{ij} - \hat{C}_{ij}$ for each candidate $j$. Then:
\begin{equation}
\label{eq:rm_computation}
    \hat{R}_{ik} = \max_{j \in \mathcal{A}} [\hat{U}_{ij}] - \hat{U}_{ik}
\end{equation}
The recommendation passes RM at tolerance $\epsilon$ if $\hat{R}_{ik} \leq \epsilon$. In practice, the RS can compute this efficiently over its top-$N$ candidates rather than the full activity set.



\subsection{Spatial Environment}
\label{app:spatial}

The city is a $G \times G$ grid (default $G = 18$) where each cell is assigned a zone type---residential, employment, leisure, or mixed---according to configurable probabilities. Block spacing determines the spatial scale (default 0.5~km per cell). Distance is computed as Euclidean distance scaled to kilometers.

Each zone produces a configurable number of places drawn from category distributions appropriate to its zone type. Leisure zones produce museums, parks, live music venues, and caf\'{e}s; mixed zones produce a broader category mix; residential zones produce caf\'{e}s, takeout, parks, and fitness venues. Each place is assigned a random rating (3.5--4.9), review count (20--5000), popularity score, and a keyword tuple. A required-category check ensures that all seven leisure subtypes have at least one matching POI.



\subsection{Travel Cost Estimator}
\label{app:tce}

The Travel Cost Estimator (TCE) learns each user's value of time $\hat{v}_i$ from continuous feedback via an exponential moving average (EMA) with learning rate $\eta_{\text{TCE}} = 0.15$. After each trip, the TCE observes the user's continuous satisfaction signal $s_{ik} \in [-1, 1]$ alongside the realized travel time and monetary cost. If the user reports dissatisfaction with a distant activity despite high preference match, the VOT estimate is adjusted upward by 5\%; if the user reports satisfaction with a longer trip, VOT is adjusted downward by 3\%. The asymmetry between these adjustments reflects the well-documented principle that negative experiences carry greater informational and psychological weight than comparably-sized positive experiences \cite{tversky1991loss, kahneman_prospect_1979, baumeister2001bad}. Dissatisfaction with a time-costly trip is a strong signal that the traveler's marginal disutility of travel time exceeds the current estimate, whereas satisfaction with a long trip is a weaker signal because it may reflect activity enjoyment rather than a genuinely low time cost. The resulting $5{:}3$ ratio lies within the range of empirically observed loss-aversion coefficients ($\approx 1.5$--$2.5$) reported in riskless-choice experiments \cite{tversky1991loss}, and it prevents the failure mode in which streams of mildly positive feedback on short trips cause the VOT estimate to collapse toward zero.

Estimates are bounded between \$5/hr and \$50/hr. The TCE also maintains a per-user uncertainty estimate $\sigma_i = \max(0.15, 1/\sqrt{n_i})$ that decreases with the number of observations $n_i$, and is used in PUP computation via Eq.~\ref{eq:pup_computation}. 

\subsection{Willingness-to-Accept Model}
\label{app:eta}

The dynamic willingness-to-accept parameter $\eta_i$ is computed as follows:
\begin{equation}
\label{eq:eta}
    \eta_i = \text{clamp}\!\left(\eta_i^{\text{base}} + \Delta^{\text{qual}} + \Delta^{\text{mem}},\; 0.05,\; 0.95\right),
\end{equation}
where $\eta_i^{\text{base}}$ is a baseline willingness calibrated from the agent's trust in platforms and autonomy preference, $\Delta^{\text{qual}}$ scales linearly with the RS's own score for the suggested place, and $\Delta^{\text{mem}}$ increases with the count of previously accepted recommendations. Concretely,
\begin{align*}
    \eta_i^{\text{base}} &= 0.52 + 0.40\,(T_i - 0.5) - 0.40\,(A_i - 0.5), \\
    \Delta^{\text{qual}} &= 0.18\,(S_k - 0.5), \\
    \Delta^{\text{mem}}  &= 0.015 \cdot \min(5,\; n_{\text{accepted}}),
\end{align*}
where $T_i$ and $A_i$ are the agent's trust-in-platforms and autonomy-preference latent variables (both $\in [0,1]$), $S_k$ is the normalised RS score of the suggested place, and $n_{\text{accepted}}$ is the agent's cumulative count of accepted recommendations (capped at five). The clamp bounds $[0.05,\,0.95]$ preserve a floor of exploration and a ceiling of skepticism.

\section{Agent Model}
\label{app:agent_model}

This appendix collects implementation details for the agents themselves: the cross-trait profile generator that produces each agent's demographics (\S\ref{app:synthetic_pop_gen}), the value of time calibration process (\S\ref{app:value_of_time}), and agents' psychological traits (\S\ref{app:psychological_traits}).

\subsection{Synthetic Population Generation}
\label{app:synthetic_pop_gen}

We construct a synthetic population using an orthogonal-factor design that minimizes cross-factor correlations \cite{box_statistics_2005}. This allows clear attribution of variation in recommendations to specific persona attributes. Each synthetic persona is defined over $18$ factors relevant to local leisure decisions, including sociodemographics (age, biological sex, income), home location (e.g., urban, suburban, rural), mobility resources (car/transit access), budget, time window preference, group composition (e.g., solo, with friends, with kids), among others. Each factor has 2-5 discrete levels, yielding a large combinatorial space from which we sample $250$ highly diverse personas. Table \ref{tab:factors} summarizes the full factor design used to generate our synthetic personas. 

\begin{table}[t]
\centering
\caption{Persona factor design and number of levels.}
\begin{tabularx}{\textwidth}{l|X|l}
\toprule
\textbf{Factor} & \textbf{Options} & \textbf{Levels} \\
\midrule
Age & \{teen, young adult, mid adult, old adult, senior\} & 5 \\
Biological Sex & \{male, female\} & 2 \\
Income & \{ $\leq$ \$35k, \$35-75k, \$75k-125k, \$125k+ \} & 4 \\
Home location & \{urban core, inner suburb, outer suburb, rural fringe\} & 4 \\
Car access & \{own car, carshare, no car\} & 3 \\
Transit access & \{high, medium, low\} & 3 \\
Budget (for an outing) & \{low, medium, high\} & 3 \\
\bottomrule
\end{tabularx}
\label{tab:factors}
\end{table}

\subsection{Value of time calibration}
\label{app:value_of_time}

Each agent's value of time is computed as $v_i = \gamma_w \cdot (\text{income}_i / H)$, where $\text{income}_i$ is the agent's annual income, $H = 2080$ is the assumed annual work hours, and $\gamma_w = 0.5$ is a scaling factor reflecting that leisure-trip travel time is typically valued at a fraction of the market wage rate. This choice is grounded in a large transportation-economics literature: Small's \cite{small2012valuation} review of empirical VOT estimates finds that the mean value of time for commuting and personal travel clusters around 50\% of the gross wage, and the U.S. Department of Transportation's guidance for benefit-cost analysis explicitly recommends valuing personal local travel at 50\% of the hourly median household income \cite{usdot_vot_2016}. Yan et al. \cite{yan_valuing_2025} further show that LLM-based synthetic travelers reproduce this $\approx$50\% VOT benchmark across a range of income and trip-purpose contexts. For an agent earning \$60{,}000/year, the resulting $v_i \approx \$14.4$/hr is therefore consistent with both empirical estimates and official practice. The parameter $\gamma_w$ can be adjusted to reflect different assumptions about the marginal utility of leisure time.

\subsection{Psychological trait model}
\label{app:psychological_traits}

Each agent carries a Big Five personality vector (openness, conscientiousness, extraversion, agreeableness, neuroticism) and four latent variables: maximization tendency, trust in platforms, autonomy preference, and algorithmic awareness. 

The latent traits modulate ten behavioral coefficients via a parameterized mapping: platform affinity, autonomy guard, planning orientation, social orientation, variety seeking, risk aversion, budget sensitivity, trend susceptibility, AI affinity, and feedback loop strength. Each coefficient is computed as a weighted sum of the relevant personality traits and latent variables, following the integrated choice and latent variable (ICLV) approach \cite{vij_how_2016}, which embeds psychometric indicators inside discrete-choice models so that unobserved attitudinal constructs can influence behavioral parameters in a theoretically grounded way. The behavioral coefficients, in turn, shift the agent's operational decision parameters: preference weights, attitude strengths, loss aversion, satisficing threshold, and the dynamic willingness to accept recommendations ($\eta$).

\section{Robustness Checks}
\label{app:robustness}

To assess how sensitive the paper's main conclusions are to parameter choices, we run a set of one-factor-at-a-time (OFAT) sensitivity sweeps grouped into three tasks. Set~1 varies knobs that govern the welfare-criteria uncertainty state (the belief standard deviation $\sigma_i$ used inside PUP and the TCE that learns each user's value of time). Set~2 varies knobs that govern how the base recommender stack learns from user feedback (personalization weights, keyword discount, feedback-strength clamps). Set~3 varies knobs that shape the base stack itself (GMaps prominence/relevance/proximity weights) and the synthetic place catalog it reads from (rating, review-count, popularity-multiplier distributions). Across the three tasks we run 35 single-knob perturbations.

Each configuration reuses the paper's full-scale simulation with a matched-seed paired-comparison design across five random seeds. For each configuration we re-run every treatment whose output the perturbed knob can affect. Reported cell means are averages over the five seeds; reported standard deviations are between-seed SDs. Paired-comparison metrics (\% harmed, Over-Recommendation Cost) are computed per seed against the matched No~RS counterfactual at the baseline configuration and averaged across seeds; absolute statistics (mean utility, negative-utility rate, Gini) are per-run statistics averaged across seeds directly.

Across the 35 single-knob perturbations,
the paper's core ordering (No~RS $<$ Standard~RS $<$ PUP $\approx$ RM $<$ Oracle) is preserved in 34 cases. The single exception is Set~1's \emph{K2-hi} (uncertainty-decay numerator set to 2.0, i.e. $\sigma$ inflated by a factor of two): RM's mean per-trip utility collapses to $+0.0646$, \emph{below} Standard~RS's $+0.0831$, and its percent-harmed rate rises to 38.96\% vs Standard RS's 38.0\%. This is RM's known failure mode: the $\varepsilon$-regret band compares a point estimate of regret against a fixed tolerance, so doubling $\sigma$ shifts many beneficial recommendations into the reject region. PUP-0.3 under the same perturbation still reads $+0.2284$ utility and 31.4\% harmed because its probability-of-positive framing absorbs the larger $\sigma$ through the normal CDF rather than against a hard threshold. The remaining 34 perturbations leave both welfare-constrained policies strictly above Standard RS in mean utility.

\paragraph{Set 1 (uncertainty / TCE, Figure~\ref{fig:robust_task1}).} Two knobs dominate: K2 $\sigma$-inflation is the only regime-change perturbation, and K5 VOT clamp widening costs PUP about 0.04 utility by letting the empirical-Bayes implied VOT climb beyond what is reasonable, which in turn inflates the \texttt{estimate\_travel\_cost} prediction and over-filters the trips PUP is meant to evaluate.

\paragraph{Set 2 (feedback learning, Figure~\ref{fig:robust_task2}).} P2 (personalization weight) is the dominant knob: doubling it to $(0.20, 0.24)$ costs PUP 0.033 utility and leaves RM nearly unchanged, while halving it lifts both. This is the Task-2 counterpart of PUP's $\sigma$-sensitivity in Set~1: here it is PUP that degrades faster than RM as the personalization signal is amplified, because PUP is the less-restrictive gate and is therefore more exposed to the larger shortlist churn. 

\paragraph{Set 3 (RS stack / catalog, Figure~\ref{fig:robust_task3}).} GMaps proximity scale is the dominant knob: tightening to 2.0~km lifts RM to $+0.2611$ (the highest RM cell in the grid), while loosening to 10.0~km drops RM to $+0.1950$ (the lowest); PUP is nearly flat across both levels. This is the Task-3 complement to Set~2's P2 effect: here RM is the more exposed gate because the $\varepsilon$-band assumes the base ranking is well-aligned with realized utility, whereas PUP's probability-of-positive gate is less sensitive to ranking quality. S1 (GMaps prom/rel split) and S3 (Popularity weights) produce smaller asymmetric shifts in the same direction. Standard~RS remains uniformly below both welfare-constrained policies across the entire grid.

\begin{figure}[h]
    \centering
    \includegraphics[width=0.95\textwidth]{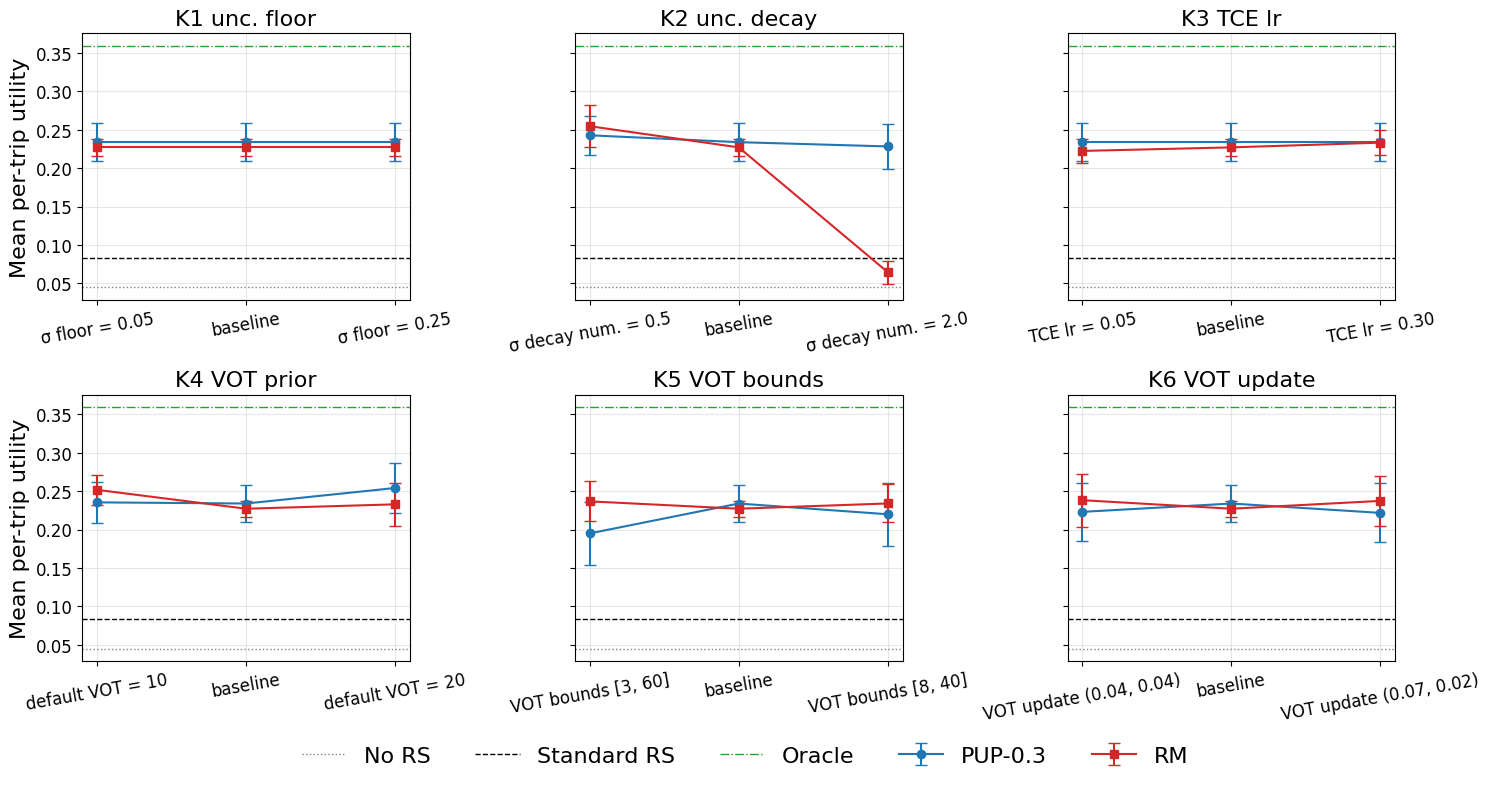}
    \caption{Set~1 OFAT sensitivity over uncertainty / TCE knobs (K1--K6). Each panel shows mean per-trip utility $\pm$ between-seed SD for PUP-0.3 (blue) and RM ($\varepsilon{=}0.27$, red) at the low, baseline, and high levels of one knob. Dotted grey line marks the No-RS reference; dashed green line marks the Oracle reference (both at baseline). Standard RS sits near 0.08 throughout and is shown as a grey series. The K2-hi cell is the only regime-change perturbation across all three tasks.}
    \label{fig:robust_task1}
\end{figure}

\begin{figure}[h]
    \centering
    \includegraphics[width=0.95\textwidth]{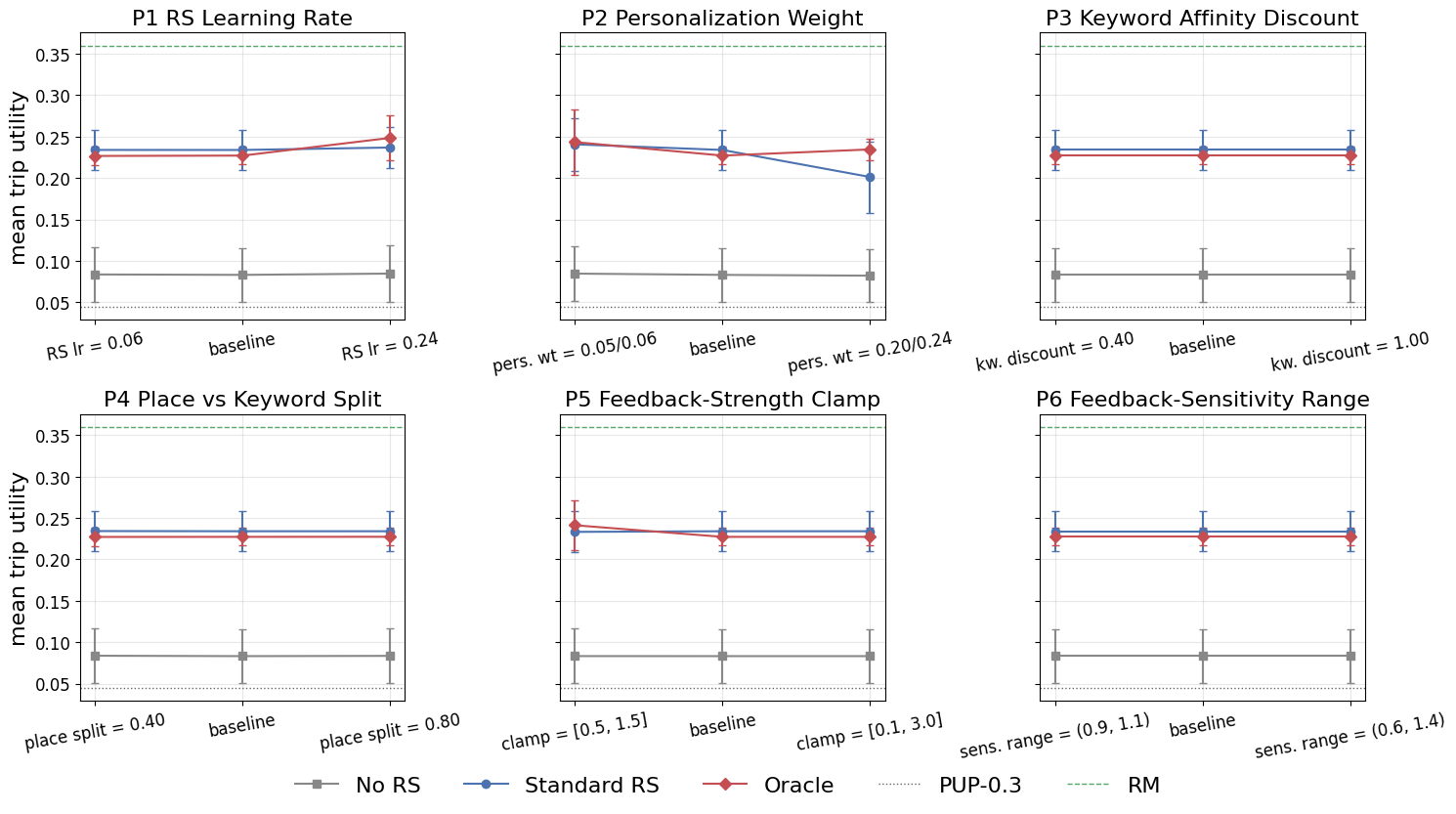}
    \caption{Set~2 OFAT sensitivity over personalization / feedback-learning knobs (P1--P6). P2 (personalization weight) is the dominant knob and is asymmetric: doubling it costs PUP-0.3 about 0.03 utility while leaving RM essentially flat.}
    \label{fig:robust_task2}
\end{figure}

\begin{figure}[h]
    \centering
    \includegraphics[width=\textwidth]{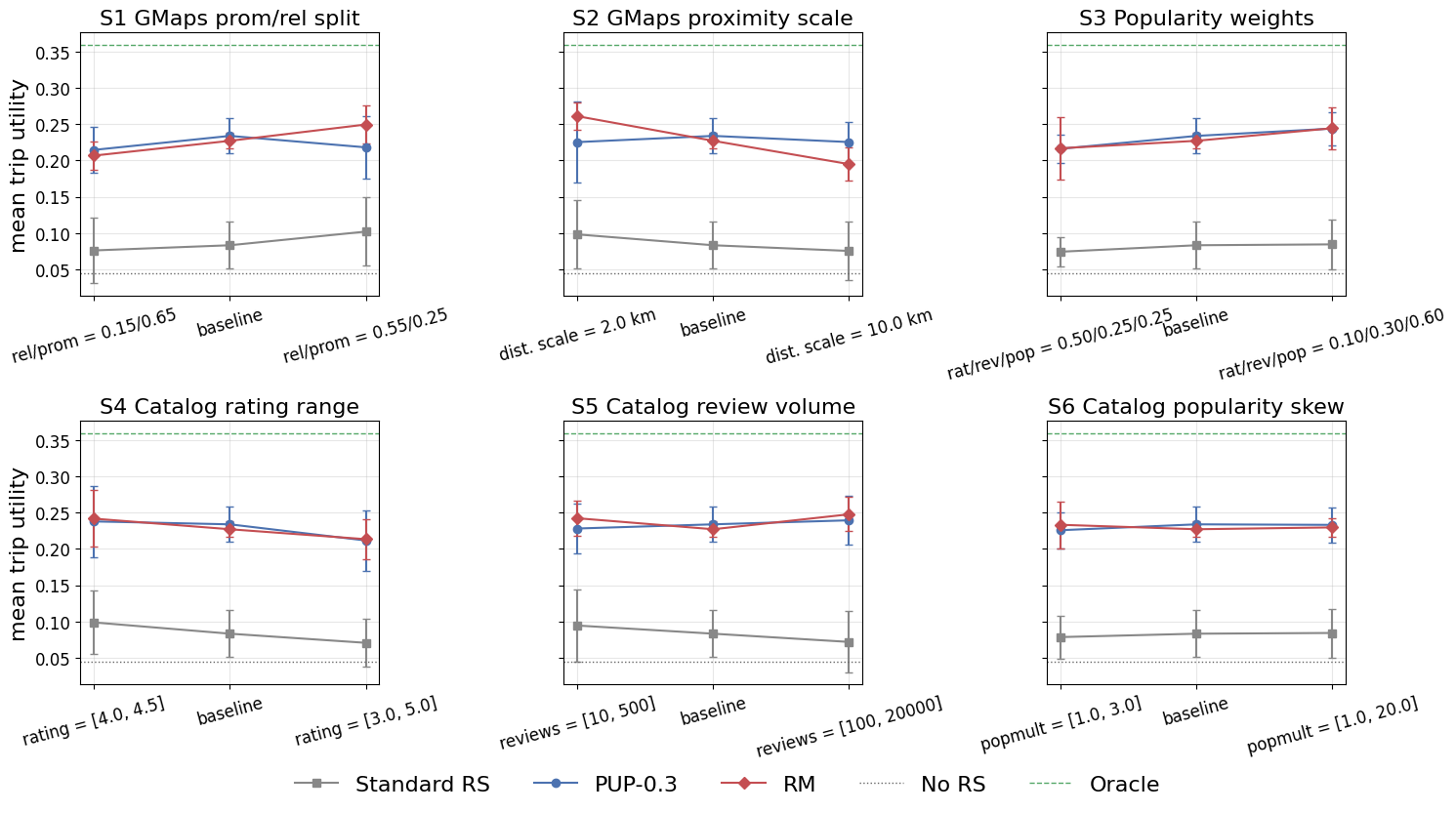}
    \caption{Set~3 OFAT sensitivity over RS-stack / catalog knobs (S1--S6). S2 (GMaps proximity scale) is the dominant knob, asymmetric in the opposite direction from Set~2: RM is the exposed policy under base-stack/geography misalignment, whereas PUP stays stable across both S2 levels.}
    \label{fig:robust_task3}
\end{figure}

\end{document}